\documentclass[11pt]{article}
\usepackage{fullpage}
\usepackage{hyperref}
\usepackage{graphicx}
\usepackage{graphicx}
\usepackage{float}
\usepackage{amsfonts}
\usepackage{amsmath,amsthm,amssymb}
\usepackage{epstopdf}
\usepackage{subfig}
\usepackage{algorithm, algorithmic}
\usepackage{color}
\usepackage{wrapfig}
\usepackage{marvosym}
\usepackage{cleveref}
\usepackage{ragged2e}
\usepackage[scaled]{helvet} 
\usepackage{fourier}
\usepackage{pdflscape}
\usepackage[english]{babel}
\usepackage{multirow}
\usepackage[usenames,dvipsnames,svgnames,table]{xcolor}

\newcommand{\E}{\mathbb{E}}
\newcommand{\nx}{\|\bx\|_2}

\newcommand{\tg}{g}

\newcommand{\bx}{\boldsymbol{x}}

\newcommand{\by}{\boldsymbol{y}}
\newcommand{\bg}{\boldsymbol{g}}

\newcommand{\bh}{\boldsymbol{h}}

\newcommand{\bm}{\boldsymbol{m}}

\newcommand{\bPhi}{\boldsymbol{\Phi}}

\DeclareMathOperator*{\argmin}{arg\,min}

\newtheorem{thm}{Theorem}[section]
\newtheorem{cor}[thm]{Corollary}

\newtheorem{rmk}[thm]{Remark}
\newtheorem{defn}[thm]{Definition}

\begin{document}

\title{Optimizing quantization for Lasso recovery}
\author{Xiaoyi~Gu, Shenyinying~Tu, Hao-Jun~Michael~Shi, Mindy~Case, Deanna~Needell, and~Yaniv~Plan% <-this % stops a space
\thanks{X.G., S.T., H.M.S., and M.C. are with Univ. of California, Los Angeles CA 90095, USA. D.N. is with Claremont McKenna College, Claremont CA 91711, USA. Y.P. is with Univ. of British Columbia, Vancouver BC V6T 1Z2, Canada.}% <-this % stops a space
}

\maketitle
\begin{abstract}
This letter is focused on quantized Compressed Sensing, assuming that Lasso is used for signal estimation. Leveraging recent work, we provide a framework to optimize the quantization function and show that the recovered signal converges to the actual signal at a quadratic rate as a function of the quantization level. We show that when the number of observations is high, this method of quantization gives a significantly better recovery rate than standard Lloyd-Max quantization. 
We support our theoretical analysis with numerical simulations.
\end{abstract}

\section{Introduction} \label{intro}

We consider the structured linear model $\by = \bPhi \bx \in \mathbb{R}^M$, $\bx \in K \subset \mathbb{R}^N$, $\bPhi \in \mathbb{R}^{N\times M}$.  Given $\by, \bPhi, K$, a goal of interest is to recover $\bx$.  Here, $K$ can encode sparsity as in the \textit{Compressed Sensing} (CS) setting, or more generally, small total variation, low-rank matrices, sparsity in a dictionary, etc.

In the noisy model ($\by \approx \bPhi \bx$), this problem can be solved by minimizing the $\ell_2$ loss function, called $K$-Lasso in \cite{planNonlinear}:
\begin{equation} \label{goal}
\text{minimize      } \|\bPhi \bx' - \by \|_2 \quad \text{       subject to       } \quad \bx' \in K.
\end{equation}

In practice, $\bPhi \bx$ must be quantized to fit in the digital domain, inducing inaccuracy in $\by$.  Work on quantization in CS includes worst case or average distortion \cite{11,Ref_SP}, reconstruction error bounds for Basis Pursuit \cite{7}, the CoSaMP method \cite{Ref_SP,5,6}, and relaxed Brief Propogation  \cite{Relaxed_BP}. The uniform quantizers \cite{Uniform,9,RefWorks:6} and standard Loyd-Max quantizers \cite{Ref_SP} are some typical examples of quantization schemes. 

Recently, Plan and Vershynin \cite{planNonlinear} analyzed the non-linear model 
\begin{equation} \label{model}
\by_i = f((\bPhi \bx)_i),~~ \forall i = 1, \dots M, \qquad  \quad \bx \in K,
\end{equation}
\noindent for some non-linear function $f$, giving tight recovery error bounds for $K$-Lasso \eqref{goal}.
In this letter, we specialize and synthesize these bounds when the non-linearity encodes quantization.
We propose to choose the quantization function which optimizes the tight error bound.  This method has also been recently proposed in \cite{thrampoulidis2015lasso}, based on a different analysis and without the focus of quantization.  
We carefully simplify and bound the error rate with the optimized quantization function, and show that often it significantly outperforms conventional quantization methods.

\section{Preliminaries and Model} \label{sec2}
\paragraph{Preliminaries} 

We use the same notation as in \cite{planNonlinear}. We say an event has ``high probability'' if it has probability at least 0.99. The notation $\lesssim$ hides an absolute constant, and $g$ ($\bg$) is a standard normal variable (vector). Throughout the paper, we assume that $\bPhi$ has independent standard normal entries,  $\by$ follows the model \eqref{model}, and $M \gtrsim d(K)$.  We utilize the \textit{mean width} and \textit{tangent cone} to give an effective measure of the dimension of $K$.  For more background on \textit{local mean width}, see \cite{planNonlinear}.
\begin{defn}
Let $B_2 \subset \mathbb{R}^N$ be the unit Euclidean ball.  The {\textsl{local mean width}} of a subset $K \subset \mathbb{R}^N$ is defined as $w(K) = \E \sup_{\bx \in K \cap B_2} \langle \bx, ~\bg\rangle.$  
%\end{defn}
%\begin{defn}
The \textsl{tangent cone} of set $K$ at $\bx$ is 
$D(K,~\bx) := \{\tau \bh : \tau \geq 0, ~ \bh \in K - \bx \}.$
\end{defn}
\begin{thm} \label{plan} \cite{planNonlinear}
Suppose $\bx \in S^{N-1}$, and $\mu \bx \in K$. Let $d(K) := w(D(K,\mu \bx))^2$.   Then with high probability, the solution $\bx'$ to problem \eqref{goal} satisfies
\begin{align} \label{bound}
\| \bx' - \mu \bx \|_2 \lesssim \left(\sqrt{d(K)} \sigma + \eta\right)/\sqrt{M}&,
\end{align}
where $\mu = \mathbb{E}[f(g)\cdot g],$  $\sigma^2 = \E[(f(g) - \mu g)^2],$ and $\eta^2 = \E[(f(g) - \mu g)^2 g^2]$.
\end{thm}

The quantity $d(K)$ plays the role of the dimension of $K$ locally near the point $\mu x$.  For example, if $K$ is an $n$-dimensional subspace then $d(K) \approx n$. 
\paragraph{Model}
We assume the model \eqref{model} with $f = f_Q$ for a quantization function $f_Q$.  We assume that $\bx \in S^{N-1}$, or equivalently that $\|\bx\|_2$ is known (and can thus be scaled); in Section \ref{robust} we relax this assumption.

We denote the \textit{quantizer} or \textit{quantization function} as 
\begin{equation}\label{fqdef}
f_Q(x) = m_i \,\,\,\text{ for } \,\,\,x \in [\tau_i, \tau_{i+1}]
\end{equation}
\noindent where $\{\tau_1 = -\infty, \tau_2, \dots, \tau_{Q+1}= \infty \}$ partitions the real line and $\{m_1, m_2, \dots, m_{Q} \}$ are the associated quantization values. 

A conventional Lloyd-Max quantizer is derived from the following idea.  The distribution of $(\bPhi \bx)_i$ is standard normal and thus known. It is natural to choose the quantizer to minimize the Mean Squared Error (MSE) function defined as 
\begin{equation} \label{opeq}
MSE := \E[(f_Q(g) - g)^2] = \sum_{i = 1}^Q \int_{\tau_i}^{\tau_{i+1}} (m_i - \tau)^2 p_g(\tau) d\tau,
\end{equation}
\noindent where $p_g$ is the probability density function of the standard normal $g$. Such a minimization problem can be solved iteratively by using the Lloyd Max algorithm \cite{LloydMax}. 

Note that the MSE only helps control the right-hand side of \eqref{bound}, but not the scaling factor $\mu$.  This is typically not one, and may thus lead to sub-optimal recovery error.  In this letter, we investigate the behavior of $K$-Lasso with non-linear measurements if $f = f_Q$ is obtained from (i) minimizing the MSE \eqref{opeq} (conventional Lloyd-Max quantization) and (ii) minimizing the MSE \eqref{opeq} with restriction to $\mu = 1$ (our proposed quantization function ).  We find that enforcing the latter restriction can significantly improve the error rate.  That is the main result of our paper.  We also note that both methods of quantization do not require knowledge of the structure, $K$, that is used for the $K$-Lasso.  Thus, the results of this paper can be useful both for the various signal structures associated with CS, and also when $x$ belongs to a linear subspace, and vanilla least-squares estimation is used.
\section{Theoretical Results} \label{sec3}
\subsection{Unit norm signals and MSE without restriction} \label{no_restriction}

When $f_Q$ is taken to minimize the MSE \eqref{opeq}, we have the following corollary.
\begin{cor}
Suppose $\bx \in S^{N-1}$ and the quantizer of the form \eqref{fqdef} is the minimizer:
\begin{equation}
f_Q = \argmin_{m_i, ~\tau_i}\,\,\,\E[(f_Q(g) - g)^2].
\end{equation}
\noindent Then with high probability, the solution $\bx'$ to problem \eqref{goal} satisfies
\begin{equation} \label{conclusionnoise}
|Q^{-2} - \frac{\sqrt{d(K)}(Q^{-2} - Q^{-4}) + 2 - Q^{-2}}{\sqrt{M}}|  \lesssim \| \bx' - \bx \|,
\end{equation}
and
\begin{equation*}
\| \bx' - \bx \| \lesssim \frac{\sqrt{d(K)}(Q^{-2} - Q^{-4}) + 2 - Q^{-2}}{\sqrt{M}} + Q^{-2}.
\end{equation*}
\end{cor}
\begin{proof}
We use the notation of Theorem \ref{plan} and explicitly compute $\mu$, $\sigma$ and $\eta$. Since $f_Q$ minimizes the MSE, we differentiate \eqref{opeq} with respect to $\tau_i$ and $m_i$ to get 
\begin{equation} \label{der1}
\int_{\tau_i}^{\tau_{i+1}} (m_i - \tau) p_{g}(\tau) \,d\tau = 0,~~~~~~\tau_i = \frac{m_i+m_{i-1}}{2},  ~~~~~~\forall i.
\end{equation}
By the definition of $\mu$ \eqref{bound} and \eqref{der1}
\begin{equation}
\mu = \sum_{i = 1}^Q \int_{\tau_i}^{\tau_{i+1}} m_i^2 p_{g}(\tau)\,d\tau = \E[f_Q(\tg)^2].
\end{equation}
Next, set 
$$e_Q := \inf_{m_i, \tau_i} MSE %\int_{-\infty}^{\infty} %\sum_{i = 1}^Q \int_{\tau_i}^{\tau_{i+1}} (m_i - x)^2 %p_{X,Y}(x,y)\,dy \,dx \\
 = \E[f_Q(\tg)^2] - 2 \E[f_Q(\tg) g] + \E[g^2]
= -\mu + 1.$$ Then, $\sigma^2 = \E[f_Q(\tg)^2] - \mu^2= \mu - \mu^2 = e_{Q} - e_{Q}^2.$ 
By the Minkowski and H\"{o}lder inequalities, %$
%\eta^2 \leq 3 (2 - e_{Q})^2.
\begin{align*}
\eta^2 %& = \E[(f_Q(g) - \mu g)^2 g^2] 
&\leq (\E[f_Q(\tg)^2 g^2]^{1/2} + \E[\mu^2 g^4]^{1/2})^2\\ 
&\leq (\E[f_Q(\tg)^4]^{1/4} \E[g^4]^{1/4}+ \sqrt{3}\mu)^2\\
&= (3^{1/4}\E[f_Q(\tg)^4]^{1/4}+ \sqrt{3}\mu)^2.
\end{align*}
 In addition, 
\begin{align} \label{eta}
%\begin{split}
\E[g^4] - \E[f_Q(\tg)^4] &= \sum_{i = 1}^Q \int_{\tau_i}^{\tau_{i+1}} (\tau^4 - m_i^4) p_g(\tau) \, d\tau \\
% &= \sum_{i = 1}^Q \int_{\tau_i}^{\tau_{i+1}} (x^2 + m_i^2)(x+m_i)(x-m_i) p_X \, dx \\
&\geq C \sum_{i = 1}^Q \int_{\tau_i}^{\tau_{i+1}} (\tau^2-m_i^2) p_g(\tau) \, d\tau \notag\\
%&= C (\E[g^2] - E[f_Q(\tg)^2]) \\
&= C e_Q \geq 0,\notag
%\end{split}
\end{align}
\noindent for some positive constant C, so we have $\E[f_Q(\tg)^4] \leq \E[g^4] = 3$. Summarizing, $\eta^2 \leq 3(\mu + 1)^2 \leq 3 (2 - e_Q)^2$. As in e.g. \cite{Quantization}, define $R(f_Q): = \log_2{Q}$, which represents the rate of quantizer coding. Then we have 
\begin{equation}\label{e_Q}
e_Q = \inf_{f_Q: R_{f_Q} \leq R} MSE \cong \frac{1}{12} (\int p_g(\tau)^{1/3} \, d\tau)^3 2^{-2R} \cong \frac{1}{12} 6 \pi \sqrt{3} \, Q^{-2}.
\end{equation}
Substituting \eqref{e_Q} into the above bounds for $\mu$, $\sigma^2$, $\eta^2$ completes the proof.

The lower bound of $\|\bx' - \bx \|$ follows from the fact that $\| \bx' - \bx \|_2 \geq |\|\mu \bx - \bx \|_2 - \| \bx' - \mu \bx \|_2 |= |~ |\mu-1| - \| \bx' - \mu \bx \|_2|$. 
\end{proof}
We see that as $Q$ increases, $\bx'$ converges quadractically to $\bx$. The constant $2$ in the numerator eventually fades out as the number of measurements increases. 

\subsection{Optimal Quantization with restriction $\mu = 1$}\label{restriction} 

We next minimize \eqref{opeq} while enforcing $\mu = 1$ to obtain a bound for $\|\bx' - \bx \|_2$ directly, and compare the results to the previous section. 
\begin{cor}
Suppose $\bx \in S^{N-1}$, and the quantizer of the form \eqref{fqdef} is the minimizer:
\begin{equation}
f_Q = \argmin_{m_i, ~\tau_i} \E[(f_Q(g) - g)^2]~~~~~~ \text{ s.t. } \mu\ = 1,
\end{equation}
\noindent with $\mu$ defined in \eqref{bound}. Then with high probability, the solution $\bx'$ to problem \eqref{goal} satisfies
\begin{equation} \label{eq: proposed method bound}
\| \bx' - \bx \| \lesssim \frac{\sqrt{d(K)} (Q^{-2} + Q^{-4}) + 2+(Q^{-2} + Q^{-4})}{\sqrt{M}}.
\end{equation}
\end{cor}
\begin{proof}
To solve this optimization problem, we use Lagrange Multipliers to solve $\nabla MSE = \lambda \nabla \mu$ with constraint $\mu = 1$. Equivalently,
\begin{equation}\label{dervar1}
\int_{\tau_i}^{\tau_{i+1}} ((\lambda + 2) \tau - 2 m_i )p_{g}(\tau)\,d\tau= 0,~~~~~~\tau_i = \frac{m_i + m_{i-1}}{\lambda+2}, ~~~~~~\forall i.
\end{equation}
Then similar to the computations in Section \ref{no_restriction}, we have $\sigma^2 = \E[f_Q(\tg)^2] - \mu^2$ and
\begin{equation}
1 = \mu = \sum_{i = 1}^Q \int_{\tau_i}^{\tau_{i+1}} \frac{2}{2+\lambda} m_i^2 p_{g}(\tau) \,d\tau
= \frac{2}{2+\lambda} \E[f_Q(\tg)^2].
\end{equation}
\noindent Define $e_Q':= \min_{m_i, \tau_i} MSE~$ under the constraint $\mu = 1$, then $
e_Q' = \E[f_Q(g)^2] - 2 \E[f_Q(g)g] + 1 
= \lambda/2 = \sigma^2$ and $\eta^2 \leq 3(e_{Q}' + 2)^2$ by similar compututations as in Section \ref{no_restriction}.

Next we analyze the relationship between $e_Q$ and $e_Q'$.
Let $\bm = (m_1, \dots, m_{Q})$ and let the optimal quantization levels be $\{m_i^*, ~\tau_i^*\}$. Treat the $\text{MSE} = \text{MSE}(m_i,~\tau_i^*) $ and $\mu = \mu(m_i,~\tau_i^*)$ as functions of $\{m_i\}$ evaluated at $\{\tau_i^*\}$. It suffices to find the local Lipschitz constants for $\mu$ and the MSE near $\{m_i^*\}$ , then 
\begin{equation}
\left\| \frac{d \mu}{d \bm} \right\|_1 = \sum_{i=1}^{Q} \left|\frac{d \mu}{d m_i}\right|
= \sum_{i=1}^{Q} \frac{1}{2 \sqrt{2 \pi}} \, |\int_{\tau_i^*}^{\tau_{i+1}^*} \tau e^{-\tau^2/2} \, d\tau |\gtrsim \tilde{C},
\end{equation}
\noindent for some $\tilde{C}>0$. Define $\delta := e_Q/ \| \frac{d \mu}{d \bm} \|_1 \leq e_Q/\tilde{C}$. Then by the continuity of $\mu$, there exists $\{m_i'\}$ that lies inside the $\delta-$ball of $\{m_i^*\}$ such that $\mu(m_i', ~\tau_i^*) = 1$. Next,
\begin{align*}
\left\|\frac{d \text{MSE}}{d \bm}\right\|_1 
%= 2\sum_{i=1}^Q |\int_{\tau_i}^{\tau_{i+1}} (m_i - x) p_X \, dx| 
 &\leq 2 \sum_{i=1}^{Q} \left|\int_{\tau_i^*}^{\tau_{i+1}^*} (m_i^* - \tau) p_g \, d\tau\right|+  2\sum_{i=1}^{Q}\left|\int_{\tau_i^*}^{\tau_{i+1}^*} \delta p_g \, d\tau\right|\\ 
%\\ & = 2 \sum_{i=1}^Q|\int_{\tau_i}^{\tau_{i+1}} \delta p_X \, dx| 
&\leq 2 \delta, 
% & \leq 2 \int_{-\infty}^{\infty} |\delta p_X| \, dx = 2 \delta
\end{align*}
\noindent gives $|e_{Q}' - e_{Q}| \leq 2 \delta^2 = 2 e_{Q}^2/\tilde{C}^2$. Applying \eqref{e_Q}, we have $
e_{Q}'\lesssim Q^{-2} + Q^{-4} .$ Substituting this expression for $e_{Q'}$ into the above bounds on $\mu$, $\sigma^2$, and $\eta^2$ gives the desired result.
\end{proof}

\begin{rmk}[Comparison of proposed method to standard Lloyd-Max quantization]
Comparing the error bound from our proposed method \eqref{eq: proposed method bound} and those of Lloyd-Max quantization \eqref{conclusionnoise}, we see that the former is proportional to $1/\sqrt{M}$ whereas the latter has a term which does not decrease with $M$.  Thus, as the number of observations $M$ increases, the proposed method gives much more accurate recovery.
\end{rmk}
\subsection{Quantization robustness} \label{robust}

We next consider the case where $\| \bx\|_2$ is approximately 1 and study the robustness of the Lloyd-Max quantizer to the unit norm assumption. The following is a simple corollary of Theorem \ref{plan} which removes the assumption that $\| \bx \|_2 = 1$ by rescaling. 
\begin{cor} \label{plancordelta} 
Assume that $\mu \bx \in K$ and let $d(K) := w(D(K,\mu \bx))^2$. Then with high probability, the solution $\bx'$ to problem \eqref{goal} satisfies
\begin{equation}
\left\| \bx' - \mu_{p} \frac{\bx}{\|\bx\|_2} \right\|_2 \lesssim \frac{\sqrt{d(K)} \sigma_{p} + \eta_{p}}{\sqrt{M}}
\end{equation}
\noindent where $\mu_{p} := \E[f_Q(\nx g) g]$,$\sigma_{p}^2 := \E[(f_Q(\nx g) - \mu_{p} g)^2]$, $\eta_{p}^2 := \E[(f_Q(\nx g) - \mu_{p}g)^2 g^2]$.
\end{cor} 
Assume the model $1- \delta \leq \|\bx\|_2 \leq 1 + \delta$ for $\delta$ small. As before, let $f_{Q(1)}$ be the optimal quantizer obtained from minimizing the MSE $\E[(f_{Q}(g) - g)^2]$.  Our result shows that the recovery rate is linearly proportional to the perturbation and is inversely proportional to the quantization level with quadratic rate.

\begin{cor}
Suppose $1-\delta \leq \nx \leq 1 + \delta$, and the quantizer of the form \eqref{fqdef} is the minimizer:
\begin{equation}
f_{Q(1)} = \argmin_{m_i, ~\tau_i} \E[(f_Q(g) - g)^2].
\end{equation}
\noindent Then with high probability, the solution $\bx'$ to problem \eqref{goal} satisfies
\begin{equation} %\label{conclusionnoise}
\|\bx' - \bx\|_2 \lesssim \frac{\sqrt{d(K)} (Q^{-2} + Q^{-4} + \delta) + 2 - Q^{-2} + \delta}{\sqrt{M}} + Q^{-2} + \delta.
\end{equation}
\end{cor}
%\begin{proof}
\noindent\textit{Proof sketch.} 
Without loss of generality, assume $\|\bx\|_2 = 1+\varepsilon$, where $0<\varepsilon < \delta$.  Then,
\begin{equation*}
\mu_p = \E[f_{Q(1)}((1+\varepsilon) g) g] = \sum_{i=1}^Q \int_{\tau_i}^{s_i} m_i \tau p_g d\tau + \int_{s_i}^{\tau_{i+1}} m_{i+1} \tau p_g d\tau
\end{equation*}
\noindent where $s_i = \tau_{i+1}/(1+\varepsilon) = \tau_{i+1}(1-\varepsilon) + \mathcal{O}(\varepsilon^2)$. We can bound the difference $|\mu-\mu_p|$ as
\begin{align*}
%\begin{split}
|\mu - \mu_p| &\leq \sum_{i=1}^Q \int_{s_i}^{\tau_{i+1}} |m_{i+1}-m_i| \tau p_g \, d\tau\\
%= \sum_{i=1}^Q |m_{i+1}-m_i|  \frac{1}{\sqrt{2 \pi}} (e^{-\tau_{i+1}^2 (1-\varepsilon)^2/2} - e^{-\tau_{i+1}^2/2} )+ \mathcal{O}(\varepsilon^2)\\
%&= \sum_{i=1}^Q |m_{i+1}-m_i|  \frac{1}{\sqrt{2 \pi}} \tau_{i+1}^2 e^{-\tau_{i+1}^2/2} \varepsilon + \mathcal{O}(\varepsilon^2)\\
&= \frac{\varepsilon}{\sqrt{2 \pi}}  \sum_{i=1}^Q |m_{i+1}-m_i| \tau_{i+1}^2 e^{-\tau_{i+1}^2/2} + \mathcal{O}(\varepsilon^2) \\
&\lesssim \frac{\varepsilon}{\sqrt{2 \pi}}  \sum_{i=1}^Q |\tau_{i+1}|^3 e^{-\tau_{i+1}^2/2} + \mathcal{O}(\varepsilon^2).
%\end{split}
\end{align*}
Since Gaussian functions lie in the Schwartz space, this sum converges absolutely. Thus, $|\mu - \mu_p| = \mathcal{O}(\varepsilon)$, implying that $|\mu_p - (1+\varepsilon)| \leq |\mu - 1| + |\mu-\mu_p| + \varepsilon = e_Q + \mathcal{O}(\varepsilon)$.

Next, since $\sigma^2 = \E[f_{Q(1)}(g)^2] - \mu^2$, $\sigma_p^2 = \E[f_{Q(1)}((1+\varepsilon)g)^2] - \mu_p^2$, it suffices to find an upper bound for $U_\sigma=|\E[f_{Q(1)}(g)^2]-\E[f_{Q(1)}((1+\varepsilon)g)^2]|$. By a similar argument, 
\begin{align} \label{sigmap}
%\begin{split}
U_\sigma %&= \sum_{i=1}^Q \int_{s_i}^{\tau_{i+1}} |m_{i+1}^2-m_i^2| p_X \, dx\\
%&= \frac{1}{2} \sum_{i=1}^Q |(m_{i+1}+m_i)(m_{i+1}-m_i)| \left(\Erf\left(\frac{\tau_{i+1}}{\sqrt{2}}\right) - \Erf\left(\frac{\tau_{i+1} (1-\varepsilon)}{\sqrt{2}}\right)\right)\\
%&= \frac{1}{2} \sum_{i=1}^Q |(m_{i+1}+m_i)(m_{i+1}-m_i)| \frac{2 \tau_{i+1} e^{-\tau_{i+1}^2} \varepsilon}{\sqrt{\pi}} + \mathcal{O}(\varepsilon^2)\\
%&\leq \varepsilon \sum_{i=1}^Q (|m_{i+1}| + |m_i|)^2 \frac{ \tau_{i+1} e^{-\tau_{i+1}^2}}{\sqrt{\pi}} + \mathcal{O}(\varepsilon^2)\\
&\leq \frac{4 \varepsilon}{\sqrt{\pi}} \sum_{i=1}^Q |\tau_{i+1}|^3 e^{-\tau_{i+1}^2}+ \mathcal{O}(\varepsilon^2) = \mathcal{O}(\varepsilon),
%\end{split}
\end{align}
and,
\begin{align*}
|\sigma^2 - \sigma_p^2| &\leq |\E[f_{Q(1)}(g)^2]-\E[f_{Q(1)}((1+\varepsilon)g)^2]| + |\mu^2 - \mu_p^2| \\
&= \mathcal{O}(\varepsilon) + |\mu+\mu_p||\mu-\mu_p| \\
&= \mathcal{O}(\varepsilon).
\end{align*}
Finally, $\eta_p \leq \sqrt{3}(2-e_Q) + \mathcal{O}(\varepsilon)$, since $U_\eta := |\E[f_Q((1+\varepsilon)g)^4] - \E[f_Q(g)^4]|$ is bounded as 
\begin{align*}
U_\eta %\\ %&= \sum_{i=1}^Q \int_{s_i}^{\tau_{i+1}} |m_{i+1}^4 - m_i^4| p_X \, dx
&= \frac{\varepsilon}{\sqrt{\pi}} \sum_{i=1}^Q |m_{i+1}^4 - m_i^4|\tau_{i+1} e^{-\tau_{i+1}^2} + \mathcal{O}(\varepsilon^2)\\
&\leq \frac{16\varepsilon}{\sqrt{\pi}} \sum_{i=1}^Q |\tau_{i+1}^5| e^{-\tau_{i+1}^2} + \mathcal{O}(\varepsilon^2) \\
&= \mathcal{O}(\varepsilon).
\end{align*}
\noindent Corollary \ref{plancordelta} and taking $e_Q \simeq Q^{-2}$ yields the desired result. 
\qed

\section{Numerical Experiments} \label{sec4}
Figure \ref{fig:f1} plots reconstruction errors under the assumption that $\|\bx\|_2 = 1$. All quantizers are computed using the Lloyd-Max algorithm \cite{LloydMax}. We display the error $\|\bx' - \bx\|_2$ and normalized error $\|\frac{\bx'}{ \|\bx'\|_2} - \bx \|_2$ for each reconstruction method. The dimension $N$ of the signal $\bx$ is 200 
and the number of measurements is $M = 50000$. 
Observe that the $K$-Lasso with restriction to $\mu = 1$ gives much better reconstruction than that with no restriction. 

Figure \ref{fig:f2} compares the reconstruction error of signals with unit norm and perturbed norms (1.05 here) under the same quantization levels. We only consider the recovery error of the type $\|\bx' - \bx\|_2$ for simplicity.
\begin{figure}[!tbp]
  \centering
  \subfloat[Unit Norm]{\includegraphics[width=0.5\textwidth]{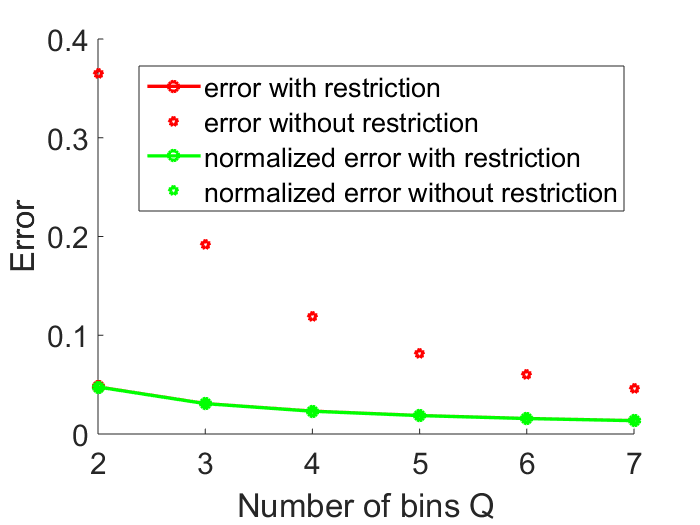}\label{fig:f1}}
  \hfill
  \subfloat[Perturbed Norm]{\includegraphics[width=0.5\textwidth]{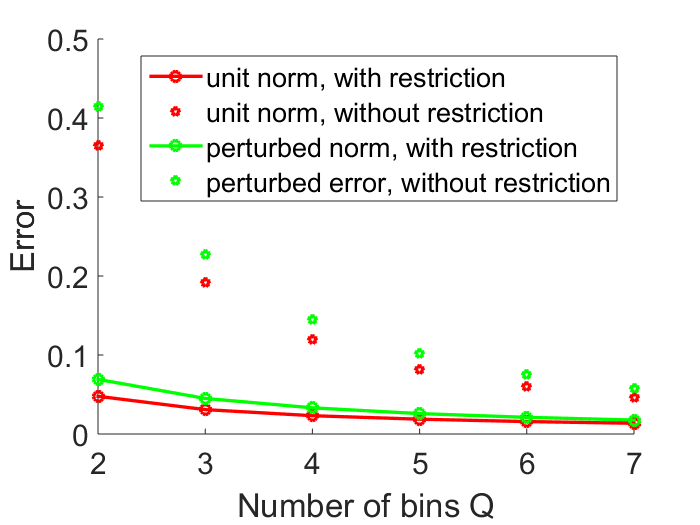}\label{fig:f2}}
  \caption{Left: Plot of recovery error vs number of bins for $k = 7$, $M = 50000$ and $N = 200$. The first, third and last lines overlap. Right: The recovery of unit norm signal and perturbed norm signal.
}\label{fig}
\end{figure}
\section{Conclusion} \label{sec5}
This letter extends existing work on the non-linear Lasso problem \cite{planNonlinear} to quantized Compressed Sensing with two different assumptions on the signal norm. When the signal norm is known, we show that the recovered signal converges to the actual signal with a quadratic rate in the quantization level. We also show that the quantizer obtained from restricting $\mu = 1$ gives a better recovery rate than the conventional Lloyd-Max quantizer. When the norm is slightly perturbed, we show that the recovery rate of the conventional Lloyd-Max quantizer is inversely proportional to the level of quantization with quadratic rate, and also linearly proportional to the degree of perturbation.

\section{Acknowledgement} \label{sec6}
This work was supported by NSF CAREER $\#1348721$, NSF DMS $\#1045536$, NSERC 22R23068, and the Alfred P. Sloan Foundation.

%\section*{References}
\bibliography{mybibfile}

\begin{thebibliography}{10}
\expandafter\ifx\csname url\endcsname\relax
  \def\url#1{\texttt{#1}}\fi
\expandafter\ifx\csname urlprefix\endcsname\relax\def\urlprefix{URL }\fi
\expandafter\ifx\csname href\endcsname\relax
  \def\href#1#2{#2} \def\path#1{#1}\fi

\bibitem{planNonlinear}
Y.~Plan, R.~Vershynin, The generalized lasso with non-linear observations,
  submitted (2015).

\bibitem{11}
V.~Goyal, A.~Fletcher, S.~Rangan, Compressive sampling and lossy compression,
  IEEE Signal Proc. Mag. 25 (2008) 48--56.

\bibitem{Ref_SP}
W.~Dai, H.~V. Pham, O.~Milenkovic, Quantized compressive sensing, preprint
  (2009).

\bibitem{7}
E.~J. Cand\`{e}s, J.~K. Romberg, T.~Tao, Stable signal recovery from incomplete
  and inaccurate measurements, Comm. Pure Appl. Math. 59~(8) (2006) 1207--1223.

\bibitem{5}
W.~Dai, O.~Milenkovic, Subspace pursuit for compressive sensing signal
  reconstruction, IEEE T. Inform. Theory 55 (2009) 2230--2249.

\bibitem{6}
D.~Needell, J.~A. Tropp, {CoSaMP}: Iterative signal recovery from incomplete
  and inaccurate samples, Appl. Comput. Harmon. A. 26~(3) (2009) 301--321.

\bibitem{Relaxed_BP}
U.~S. Kamilov, V.~K. Goyal, S.~Rangan, Optimal quantization for compressive
  sensing under message passing reconstruction, IEEE Int. Symp. Inform. Theory
  (2011) 390--394.

\bibitem{Uniform}
E.~J. Cand\`{e}s, J.~K. Romberg, Encoding the $\ell_p$ ball from limited
  measurements, Proc. Data Compression Conf. (DDC) (2006) 28--30.

\bibitem{9}
P.~T. Boufounos, R.~G. Baraniuk, Quantization of sparse representations, in:
  Rice Univ. ECE Dept. Tech. Report 0701., 2007.

\bibitem{RefWorks:6}
P.~T. Boufounos, R.~G. Baraniuk, 1-bit compressive sensing, in: 42nd Ann. Conf.
  Inform. Sciences and Systems (CISS), IEEE, 2008, pp. 16--21.

\bibitem{thrampoulidis2015lasso}
C.~Thrampoulidis, E.~Abbasi, B.~Hassibi, Lasso with non-linear measurements is
  equivalent to one with linear measurements, in: Advances in Neural Inform.
  Proc. Systems, 2015, pp. 3402--3410.

\bibitem{LloydMax}
S.~P. Lloyd, Least squares quantization in {PCM}, IEEE T. Inform. Theory IT-28
  (1982) 129--137.

\bibitem{Quantization}
R.~M. Gray, D.~L. Neuhoff, Quantization, IEEE T. Inform. Theory 44~(6) (1998)
  2325--2383.

\end{thebibliography}

\end{document}